\documentclass[12pt]{iopart}
\usepackage{graphicx}
\usepackage{amsthm}
\usepackage{iopams}
\usepackage{color}
\usepackage{bbm}

\newcommand{\ket}[1]{\mbox{$ | #1 \rangle $}} 
\newcommand{\bra}[1]{\mbox{$ \langle #1 | $}}

\newcommand{\ie}{\emph{i.e.}}

\newtheorem{proposition}{Proposition}[section]

\newcommand{\db}{\ensuremath{\mathrm{db}}}
\newcommand\T{\rule{0pt}{3.2ex}}
\newcommand\B{\rule[-1.4ex]{0pt}{0pt}}

\begin{document}

\title{Detector decoy quantum key distribution}

\author{Tobias Moroder$^{1,2}$, Marcos Curty$^3$ and Norbert
  L\"{u}tkenhaus$^{1,2}$}   

\address{$^1$ Quantum Information Theory Group, Institute of
  Theoretical Physics I, and Max-Planck Research Group for Optics,
  Photonics and Information, University Erlangen-Nuremberg,
  Staudtstrasse 7/B2, 91058 Erlangen, Germany} 
\address{$^2$ Institute for Quantum Computing, University of Waterloo, 200
  University Avenue West, N2L 3G1 Waterloo, Canada}
\address{$^3$ ETSI Telecomunicaci\'on, Department of Signal Theory and
  Communications, University of Vigo, Campus Universitario, E-36310
  Vigo, Spain}

\eads{\mailto{tmoroder@iqc.ca}}

\begin{abstract}
Photon number resolving detectors can enhance the performance of many
practical quantum cryptographic setups. In this paper, we employ a
simple method to estimate the statistics provided by such a photon
number resolving detector using only a threshold detector together
with a variable attenuator. This idea is similar in spirit to that of
the decoy state technique, and is specially suited for those scenarios
where only a few parameters of the photon number statistics of the
incoming signals have to be estimated. As an illustration of the
potential applicability of the method in quantum communication
protocols, we use it to prove security of an entanglement based
quantum key distribution scheme with an untrusted source without the
need of a squash model and by solely using this extra idea. In this
sense, this \emph{detector decoy method} can be seen as a
different conceptual approach to adapt a single photon security proof
to its physical, full optical implementation. We show that in this
scenario the legitimate users can now even discard the double click
events from the raw key data without compromising the security of the
scheme, and we present simulations on the performance of the 
BB84 and the 6-state quantum key distribution protocols.   
\end{abstract}

\section{Introduction}

Among the research performed nowadays in order to increase the secret key rate
and distance that can be covered by quantum key distribution (QKD) systems,
one can distinguish three main work areas which are closely related to each
other \cite{qkd,qkd_dusek,qkd_valerio}. On the one hand, we have the
development of new proof techniques, together with better classical
post-processing protocols, that are able to further extend the proven
secure regimes for idealized QKD schemes, typically based on the
transmission of two-level quantum systems (qubits)
\cite{mayers96a,mayers01b,lo99b,shor00a,chau02a,gottesman03a,devetak05a,
  koashi06z, kraus05a,renner05a}. On the other hand, we find the continuous 
improvements that come from the technological side. Especially, the
design of better light sources and better detectors should give us
provable secure communications over a growing distance
\cite{stucki02a,takesue07a,rosenberg07a}. Finally, we have the 
research which aims to close the gap between theoretical security
concepts for idealized QKD schemes and their experimental realization
\cite{hwang03a,lo05a,wang05a,gottesman04a,inamori07a,scarani04a,kraus07a}.    

The awareness of such a theory-experiment gap was triggered by the important
deviations present in practical QKD setups with respect to their
original theoretical proposal, which usually demands technologies that
are beyond our present experimental capability. Especially, the signal
states emitted by the source, instead of being single photons, are
usually weak coherent pulses which can contain more than one
photon prepared in the same polarization state. Now, the eavesdropper
(Eve) is no longer limited by the no-cloning theorem \cite{wooters82a},
since the multiphoton pulses provide her with perfect copies of the
single photon. In this scenario, she can perform the so-called {\it
photon-number splitting} attack \cite{huttner95a,brassard00a}. This
attack gives Eve full information about the part of the key generated from the
multiphoton signals, without causing any disturbance in the signal
polarization. The use of weak coherent pulses jeopardizes the security of QKD
protocols, and lead to limitations of rate and distance that can be
achieved by these techniques. For instance, it turns out that the BB84
protocol \cite{bennett84a} with weak coherent pulses can give a key
generation rate of order $O(\eta^2)$ \cite{gottesman04a,inamori07a},
where $\eta$ denotes the transmission efficiency of the quantum channel.  
 
A significant improvement of the secret key rate can be obtained when the
hardware is slightly modified. In particular, by using the so-called decoy
state method \cite{hwang03a,lo05a,wang05a}. In this approach, the
sender (Alice) varies, independently and at random, the mean photon
number of each signal state sent to the receiver (Bob) by employing
different intensity settings. Eve does not know the mean photon number
of each signal sent. This means that the gain and the quantum bit
error rate (QBER) of each signal can only depend on its photon number
but not on the particular intensity setting used to generate it. From
the measurement results corresponding to different intensity settings,
it turns out that the legitimate users can estimate the gain and the
QBER associated to each photon number state and, therefore, obtain a
better estimation of the behavior of the quantum channel. This
translates into an enhancement of the resulting secret key rate. The
decoy state technique has been successfully implemented in several
recent experiments \cite{rosenberg07a,zhao06a,schmitt07a,dynes07a,yuan07a},
and it can deliver a key generation rate of the same order of
magnitude like single photon sources, \ie, $O(\eta)$
\cite{hwang03a,lo05a,wang05a}.     

The use of photon number resolving (PNR) detectors instead of
threshold detectors can also enhance the performance of many practical
QKD setups. For instance, in those situations where the decoy state
method cannot be easily applied. This is the case, for example, in a
QKD scheme with an untrusted source where the legitimate users cannot
control the mean photon number of the signal states emitted. In these
scenarios, it might be still very useful for the legitimate users to
have access to the photon number statistics of the incoming
signals. To simplify the security analysis, it is very tempting to
assume a squash model for Alice's and Bob's detection setup
\cite{PDC_qkd}. This model maps each incoming signal to a one-photon
polarization space followed by a measurement in this smaller
dimensional Hilbert space. The squash model has been recently proven
to be correct for the case of the BB84 protocol
\cite{squash1,squash2}. However, in Ref.~\cite{squash2} it was shown
that the same does not hold, for instance, for the active basis choice
measurement in the 6-state protocol \cite{bruss98a}.   

In this paper, we analyze a simple method to estimate the photon
number statistics provided by a PNR detector using only a practical
threshold detector together with a variable attenuator. 
The basic idea consists in measuring the incoming light field with a
set of simple threshold detectors with different efficiencies and thus
one obtains more information about the underlying distribution of the photons. 
This technique has its origin in the field of quantum metrology as
discussed in Refs.~\cite{mogli98,rossi04a,zambra06a} and has been
successfully implemented in some recent experiments \cite{zambra05a,
genovese06a,genovese08a} which show the practical feasibility of the method.  
Here we apply it for the first time to various realistic QKD
scenarios. For instance, it can be used to prove security of those QKD
setups that do not have a squash model \cite{squash2}, or in those
security proofs that only require the statistics given
by a PNR detector \cite{tamaki06suba}. 
However if one likes to employ this technique in the QKD context one
has to estimate the photon number statistics under the worst case
assumption for Alice and Bob. Thus the known reconstruction method
from Refs.~\cite{rossi04a,zambra06a} which considers only a truncated
version of the problem (under the additional constraints of only a
finite, small number of experimental runs) cannot be directly used for
the QKD setups. Nevertheless, the central idea of the problem remains
unchanged. In fact, this method can be considered as the decoy state
technique applied to the detector side: If Alice and Bob vary,
independently and at random, the detection efficiency of their
apparatus then they can estimate the photon number statistics of the
signals received. Note that the photon number distribution of the
incoming signals cannot depend on the particular efficiency setting
used to measure them. Therefore, from now on, we shall refer to this estimation
procedure as \emph{detector decoy} to emphasize its connection and
applicability to QKD. The detector decoy idea can be employed both for
calibrated and uncalibrated devices \cite{qkd_valerio}. The essential
requirement here is that Eve cannot modify the variable attenuator
employed by the legitimate users to vary the detection efficiency of
their setups. 

Specifically, we apply the detector decoy method to two different QKD
scenarios. In the first one, we prove the security of an entanglement
based QKD scheme with an untrusted source solely by using this
estimation procedure. More precisely, we investigate the situation
where Alice and Bob perform either the BB84 or the 6-state protocol,
and we compare the resulting key rates with those arising from a
security proof based on the squash model assumption \cite{PDC_qkd}. In
contrast to this last scenario, now Alice and Bob can now even sift
out the double click events without compromising the security of the
scheme. Note, however, that we compare different situations, since
they require different detection setups. As a second potential
application, we analyze an alternative experimental technique, also
based on the detector decoy method, to estimate the photon number
statistics of the output signals in a ``Plug \& Play'' configuration
\cite{stucki02a,muller97a,ribordy00a}.      

The paper is organized as follows. In Sec.~\ref{seca} we describe in detail
the detector decoy idea to estimate the photon number statistics of an
optical signal by means of a threshold detector combined with a variable
attenuator. Next, we apply this method to different practical QKD
scenarios. In particular, Sec.~\ref{app} analyzes the security of an
entanglement based QKD scheme with an untrusted source. Then, in
Sec.~\ref{sec_pp} we propose an experimental technique to estimate the photon
number statistics of the output signals in a ``Plug \& Play''
configuration. Finally, Sec.~\ref{CONC} concludes the paper with a summary.

\section{Estimating photon number statistics}\label{seca}

Most of the security proofs for QKD only require the estimation of a
few parameters related with the photon number statistics of the
incoming signals. These parameters suffice to obtain good lower bounds
for the achievable secret key rate. Here we discuss and explain the
technique to measure the photon number distribution of an optical
signal by means of a practical threshold detector in combination with
a variable attenuator. As mentioned this idea has been introduced 
previously in the scientific literature before,
cf. Refs.~\cite{mogli98,rossi04a,zambra06a}. The current discussion
differs in the particular way of how one reconstructs part of the
photon number distribution from the observed measurement outcomes;
here we need to provide ultimate bounds for certain photon number
parameters, cf. Prop.~\ref{prop_n}, that are valid without any further
assumptions on the signal states.
Of course, the photon number distribution can also be obtained by using
directly PNR detectors \cite{cabrera98a,kim99a}. This approach
would provide Alice and Bob not only with the distribution of
the incoming signals but also with the number of photons contained in each of
them. Unfortunately, most of the methods proposed so far in the
literature to construct this type of detectors result in devices with
low detection efficiencies and which are unable to operate at
room-temperatures. An interesting alternative is that based on
detection schemes which use, for instance, time multiplexing
techniques \cite{achilles03a,lundeen08a}. This method has allowed, for
example, a passive decoy selection in QKD,
cf. Ref.~\cite{mauerer07a,mauerer08a}. In this last case, however, the
achievable photon number resolution depends on the number of detectors
and on the number of spatially, or temporal, separate bins used.  

\begin{figure}[h!]
  \begin{center}
    \includegraphics[scale=0.7]{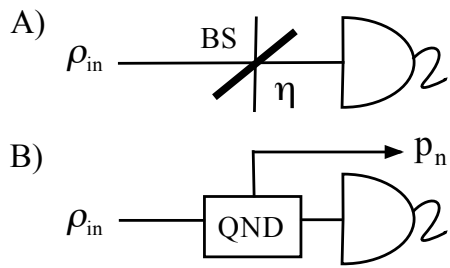}
  \end{center}
  \caption{Case A) Detection setup which combines a beam splitter of
    transmittance $\eta$ together with an ideal threshold detector. The
    incoming signal state $\rho_{\rm in}$ is given by Eq.~\ref{inputs}. Case
    B)  If one varies the transmittance $\eta$ of the beam splitter, then the
    detection setup given in Case A is equivalent to a {\it quantum
      non-demolition} (QND) measurement that provides only the photon
    number statistics of the incoming signals. \label{detector_fig}}  
\end{figure}

The basic idea of the detector decoy technique can be highlighted with
a simple example. Let the optical signals arriving to a {\it perfect}
threshold detector be either a single photon pulse or a strong pulse
containing several photons. These two signals will always produce a
single click in the detector. Therefore, in this scenario both events
cannot be distinguished. Suppose now that a beam splitter with very
low transmittance is placed before the detector and the same kind of
signals is received once more. Then, the single photon pulse will
produce much less clicks than the strong pulse. In principle, both
events can now be distinguished. That is, by varying the transmittance
of the beam splitter more information about the photon number
distribution of the signals received becomes available.    

Let $\eta$ denote the transmittance of such a beam splitter (Case A in
Fig.~\ref{detector_fig}). The combined detection setup can be
characterized by a positive operator value measure (POVM) which
contains two elements, $F_{\rm vac}(\eta)$ and $F_{\rm click}(\eta)$,
given by \cite{detector}  
\begin{equation}
  \label{detector}
  F_{\rm vac}(\eta)=\sum_{n=0}^{\infty}  (1-\eta)^n\Pi_n, 
\end{equation}
and $F_{\rm click}(\eta)=\mathbbm{1}-F_{\rm vac}$, where $\Pi_n$
represents the projector onto the $n$-photon subspace. That is, the
outcome of $F_{\rm vac}(\eta)$ corresponds to no click in the
detector, while the operator $F_{\rm click}(\eta)$ gives precisely one
detection click, which means at least one photon is detected. Suppose
for the moment that the input signal state is of the form   
\begin{equation}
  \label{inputs}
  \rho_{\rm in}=\sum_{n=0}^{\infty} p_n \rho_n,
\end{equation}
where the signals $\rho_n$ belong to the $n$-photon subspace. 

The probability of getting a click, that we shall denote as $p_{\rm
click}(\eta)$, depends on the transmittance $\eta$ of the beam
splitter. It can be calculated as  $p_{\rm click}(\eta)=\tr[F_{\rm 
click}(\eta)\rho_{\rm in}]$. Similarly, $p_{\rm vac}(\eta)=1-p_{\rm click}(\eta)$
represents the probability that the detector does not click. Using
Eq.~\ref{detector}, we find that this last quantity can be expressed as
$p_{\rm vac}(\eta)=\sum_{n=0}^{\infty}  (1-\eta)^np_n$. Now one can follow a
similar idea to that of the decoy state method. In particular, if the receiver
varies the transmittance $\eta=\{\eta_1, \ldots, \eta_M\}$ of the beam
splitter he can generate a set of linear equations with the
probabilities $p_n$ as the unknown parameters \cite{mogli98,rossi04a,zambra06a},
\begin{eqnarray}\label{lsys}
  p_{{\rm vac}}(\eta_1)&=&\sum_{n=0}^{\infty}  (1-\eta_1)^np_n, \nonumber \\
  &\vdots& \label{prob_vacuum} \\
  p_{{\rm vac}}(\eta_M)&=&\sum_{n=0}^{\infty}  (1-\eta_M)^np_n. \nonumber
\end{eqnarray}
From the observed data $p_{\rm vac}(\eta)$, together with the
knowledge of the transmittance $\eta$ used, the receiver can solve
Eq.~\ref{prob_vacuum} and obtain the value of $p_n$. For instance, in 
the general scenario where he employs an infinity number of possible
decoy transmittances $\eta\in[0,1]$, he can always estimate any finite
number of probabilities $p_n$ with arbitrary precision. This result is
illustrated as Case B in Fig.~\ref{detector_fig}. On the other hand,
if the receiver is only interested in the value of a few probabilities
$p_n$, then he can estimate them by means of only a few different
decoy transmittances, like in the decoy state method
\cite{hwang03a,lo05a,wang05a,ma05a}. This last statement is given by
Proposition \ref{prop_n} for the case where the receiver only wants to
find worst case bounds for the probabilities $p_0$, $p_1$, and
$p_2$. This proposition can straightforwardly be generalized to cover
also the case of any other finite number of probabilities
$p_n$. Note, however, that it only constitutes a possible example of
an estimation procedure that provides the exact values of the
probabilities $p_n$ in the considered limit. In principle, many other
estimation techniques are also available, like linear programming
tools \cite{linearprogramming} or different ideas from the original
decoy state method \cite{tsurumaru}.   

\begin{proposition}\label{prop_n}[Finite settings] Consider the set
  of linear equations given by 
 \begin{equation}
   f(c)=\sum_{n=0}^\infty c^n x_n,
 \end{equation} 
 where the unknown parameters $x_n$ fulfill $x_n \geq 0 $ and
 $\sum_{n=0}^\infty x_n \leq C$ for a given constant $C$, and where $c$
 satisfies $c \in [0,1]$. Consider now three different settings
 $c_0=0$, $c_1$ and $c_2$. Then, the unknown variables $x_0$, $x_1$
 and $x_2$ satisfy, respectively, $x_0=f(c_0)=f(0)$,   
 \begin{equation}
   \label{eq:bound1}
   l_1(c_1)=\frac{f(c_1)-f(0)\left( 1- c_1^2 \right) -
     c_1^2 C }{c_1-c^2_1} \leq x_1 \leq  
   u_1(c_1)=\frac{f(c_1)-f(0)}{c_1},
 \end{equation} 
 and $l_2(c_1,c_2) \leq x_2 \leq u_2(c_1,c_2)$ with bounds 
 \begin{eqnarray}
   \label{eq:l2}
   l_2(c_1,c_2)&=&\frac{f(c_2)-f(0)\left(1-c_2^3\right) -
     u_1(c_1)\left(c_2-c_2^3 \right) - c_2^3 C}{c_2^2-c_2^3},  
      \\ 
   u_2(c_1,c_2)&=& \frac{f(c_2) - f(0) -c_2 l_1(c_1)}{c_2^2}.
 \end{eqnarray}
When $c_1=\Delta$ and $c_2=\sqrt{\Delta}$, the
 given bounds converge to the exact value of the variables $x_1$ and
 $x_2$ in the limit $\Delta \to 0$.   
\end{proposition}

\begin{proof}
We present the explicit derivation of the upper bound $u_1(c_1)$ and
of the lower bound $l_2(c_1,c_2)$. The other bounds can be obtained in
a similar way. The basic idea is as follows: We first upper bound
$x_1$ from the knowledge of $f(c_1)$; afterwards this result is used
to lower bound $x_2$ given the value of $f(c_2)$. Starting with the
definition of $f(c_1)$ we obtain    
\begin{equation}
  \label{una}
  f(c_1) = x_0 + c_1 x_1 + \sum_{n=2}^\infty c_1^n x_n \geq f(0) + c_1 x_1,
\end{equation}
where we have used the fact that $x_0=f(0)$ and $x_n \geq 0 $. This inequality
already gives the upper bound $u_1(c_1)$ on $x_1$. To obtain the lower
bound $l_2(c_1,c_2)$, note that the other extra condition on the open
parameters $x_n$ gives  
\begin{equation}
  \label{dos}
  \sum_{n=N+1}^\infty x_n \leq C - \sum_{n=0}^N x_n, \;\; \forall N \in
  \mathbbm{N}.
\end{equation}
Using a similar inequality to that in Eq.~\ref{una} for the definition
of $f(c_2)$ in combination with the condition given by Eq.~\ref{dos} we obtain
\begin{eqnarray}
  \label{tres}
  \nonumber
  f(c_2) &\leq & x_0 + c_2 x_1 + c_2^2 x_2 + c_2^3\left( C-x_0-x_1-x_2
  \right) \\ 
  &\leq& f(0) (1-c_2^3) + u_1(c_1) (c_2-c_2^3) + x_2 (c_2^2-c^3_2) + c_2^3 C. 
\end{eqnarray}
In the second step we have employed again the fact that $f(0)=x_0$
together with the upper bound for $x_1 \leq u_1(c_1)$. Eq.~\ref{tres}
directly delivers the lower bound given by Eq.~\ref{eq:l2}. 

Let us now prove that both bounds converge. The unknown parameters
$x_n$ are exactly the Taylor expansion coefficients of the function
$f(c)$ evaluated at the point $c=0$, \ie, 
\begin{equation}
  x_n=\frac{1}{n!}f^{(n)}(0)=\frac{1}{n!}\frac{d^n}{d c^n} f(c) \Big|_{c=0}.
\end{equation}
This means that the upper bound $u_1(c_1)$ becomes exact if one finds the
appropriate setting to estimate the first derivative. Choosing
$c_1=\Delta$ directly gives 
\begin{equation}
  \lim_{\Delta \to 0} u_1(\Delta) = \lim_{\Delta \to 0}
  \frac{f(\Delta)-f(0)}{\Delta} = f^{(1)}(0)=x_1.
\end{equation}
For the lower bound $l_2(c_1,c_2)$ one has to perform the limit $c_1
\to 0$ prior to $c_2 \to 0$. Hence, one selects the setting
$c_2=\sqrt{\Delta}>\Delta=c_1$. Using the Taylor expansion series,
\begin{eqnarray}
  f(c_1)=f(\Delta) &\approx& f(0)+ f^{(1)}(0) \Delta, \\
  f(c_2)=f(\sqrt{\Delta}) &\approx& f(0)+ f^{(1)}(0) \sqrt{\Delta} +
  \frac{1}{2} f^{(2)}(0) \Delta,
\end{eqnarray}
we obtain
\begin{eqnarray}
  \nonumber 
  \lim_{\Delta \to 0} l_2(\Delta,\sqrt{\Delta}) &=& \lim_{\Delta \to 0}
  \frac{1}{1-\sqrt{\Delta}} \left\{ \frac{1}{2} f^{(2)}(0) +
    \sqrt{\Delta}\left[ f(0)-C+f^{(1)}(0) \right] \right\} \\ &=&
  \frac{1}{2} f^{(2)}(0) = x_2. 
\end{eqnarray}
This proves that also the lower bound on $x_2$ becomes exact
in the considered limit.
\end{proof}
\noindent For the further discussion we shall assume that one can
always obtain the exact values of the probabilities $p_n$, hence we
will ignore any finite size effects from now on. 

So far, we have analyzed the case of an ideal threshold detector. When 
the detector has some finite detection efficiency $\eta_{\rm det}$ and
shows some noise in the form of dark counts which are, to a good
approximation, independent of the incoming signals, such a detector
can be described by a beam splitter of transmittance $\eta_{\rm det}$
combined with a noisy detector \cite{yurke}. In this scenario the
argumentation presented above still holds, and the detector decoy
method can also be used in the calibrated device scenario. Note that
the operator $F_{\rm vac}(\eta)$ is now given by    
\begin{equation}
  \label{eq:noclickPOVM}
  F_{\rm vac}(\eta)=(1-\epsilon)\sum_{n=0}^{\infty}  (1-\eta\eta_{\rm
  det})^n\Pi_n,  
\end{equation}
where $\epsilon$ represents the probability to have a dark
count. In this case, $p_{\rm vac}(\eta)$ has the form 
\begin{equation}
  \label{eq:inefficient}
  p_{\rm vac}(\eta)=(1-\epsilon)\sum_{n=0}^{\infty}  (1-\eta\eta_{\rm
  det})^n{}p_n.  
\end{equation}
Again, if one varies $\eta$ then, from the measured  data $p_{\rm
  vac}(\eta)$ together with the knowledge of the parameters $\eta$,
$\eta_{\rm det}$ and $\epsilon$, the receiver can deduce
mathematically\footnote{Note that the convergence result given in 
Proposition~\ref{prop_n} does not apply directly to this scenario. 
However, it is still correct if one ignores the detector
  efficiency part. This is the case considered in Sec.~\ref{app}, when we 
  analyze the security of QKD schemes.} the value of
the probabilities $p_n$.      

The results provided in this section rely on the description of the
detectors given by Eq.~\ref{detector} and Eq.~\ref{eq:noclickPOVM}. However, 
there are many different ways to model the exact behavior
of an imperfect detector, and quite often the model is adapted
to the explicit situation for which one wants to use the calculated
data. The probability of a no click outcome given by Eq.~\ref{lsys} and
Eq.~\ref{eq:inefficient} describes the typical QKD situation quite
accurate (see, \textit{e.g.}, Ref.~\cite{nor_99}). Of course, whenever
this situation changes the exact analysis need to be
adapted. Nevertheless, the main idea behind the detector decoy method
stays invariant. Via the observations on several different input
distributions $\{ p_n(\eta) \}$, that directly depend on the incoming
photon number distribution $\{ p_n \}$ by means of an explicit, known
transformation rule (binomial transformation in the case of the beam
splitter) one can obtain more information about the incoming photon
statistics.

\section{Entanglement based QKD schemes with an untrusted
  source}\label{app}

In this section we combine the detector decoy idea with the security
statement for an entanglement based QKD scheme with an untrusted
source. The schematic setup of the experiment is shown in
Fig.~\ref{PDC_QKD}. The source, which is assumed to be under Eve's
control, is placed between the two receivers. In the ideal case, this
source produces entangled states that are sent to Alice and Bob. The
entanglement is contained in the polarization degree of freedom of the
light field. This means that at least two different optical modes have
to be considered for each side. On the receiving side, we assume that
both measurement devices only act onto these two modes. For
simplicity, we restrict ourselves to the familiar active polarization
measurement setup, in which each party actively chooses the   
measurement basis $\beta$. In the BB84 protocol each
receiver can choose between two different basis, while in the
6-state protocol all three different polarization axis can be
selected. Each measurement device consists of a polarizing beam
splitter that spatially separates the two incoming 
modes according to the chosen polarization basis $\beta$, followed 
by two threshold detectors on the two different output modes of the
beam splitter. The analysis for other measurement devices, like for
example a passive measurement setup is  completely
analogous. Entanglement based schemes constitute a very promising
alternative to implement QKD over long distances. In fact, they 
hold the theoretical distance record for a QKD scheme without quantum 
repeaters so far, cf. the simulation in Ref.~\cite{PDC_qkd}. This 
type of protocols have been successfully implemented in many different recent
experiments (See, \textit{e.g.},
Ref.~\cite{urssinnlos,ling08a,erven08} and references therein.), and
they are a suitable candidate to realize earth-satellite 
QKD links \cite{spacepaper}. For more details on the setup, or on the
measurement apparatus, we refer the reader to
Refs.~\cite{PDC_qkd,nor_99}.  
\begin{figure}[ht]
  \centering
  \resizebox{100mm}{!}{\includegraphics{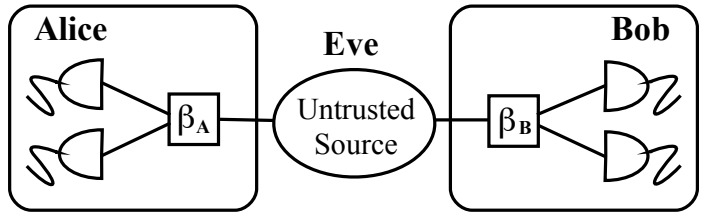}}
  \caption{Schematic diagram of an entanglement based QKD scheme with an
    untrusted source for the case of an active choice of the
    measurements basis $\beta_A, \beta_B$ respectively.}
  \label{PDC_QKD}
\end{figure}

Section~\ref{secret_rate} includes the security analysis for
an entanglement based QKD scheme. We follow the security proof technique   
provided in Ref.~\cite{kraus05a,renner05a,kraus07a}. Using this
approach allows us to 
directly pinpoint the usefulness of the detector decoy method in the
security proof. The main result of Sec.~\ref{secret_rate} is given by
Eq.~\ref{rate_mc}, which shows that the final lower bound on the
secret key rate only depends on a few essential parameters of the
system. Then, in Sec.~\ref{decoy_est} we employ the detector
decoy idea to estimate these parameters for the active measurement
scheme. Finally, Sec.~\ref{evalua} contains the simulation of a real
QKD experiment. There, we compare the detector decoy method with other
proof techniques.

\subsection{Secret key rate}\label{secret_rate}

We discuss the security level in the case of collective attacks. More
precisely, we assume that the three parties Alice, Bob and Eve share an
unlimited number of copies of the \emph{same} state $\ket{\Psi}_{\rm ABE}$. 
Once Alice and Bob have received their part of the quantum state, they
measure it to obtain information about the  state
$\rho_{\rm{AB}}=\tr_{\rm E} \left( \ket{\Psi}_{\rm ABE} \bra{\Psi}
\right)$. The POVMs used by Alice and Bob are denoted, respectively,
as $\{ F^{\rm A}_i \}$ and $\{ F^{\rm   B}_j\}$. They contain
\emph{all} the measurement operators that Alice and Bob perform during
the protocol, \ie, they also include the detector decoy
measurements. These additional measurements enable the legitimate
users to estimate some of the crucial parameters of the key rate
formula with high confidence. Let us further
assume that all the measurement operators are invariant under a
projection measurement of the total number of photons
\cite{nor_99}. That is, each
of the elements $F^{\rm A}_i$ or $F^{\rm B}_j$ satisfies   
\begin{equation}\label{mnw}
  F_i = \sum_{n=0}^\infty \Pi_n F_i \Pi_n.
\end{equation}
Using this assumption, one can consider a slightly
different, but completely equivalent, scenario for the distribution of
the quantum states. The new scenario has the advantage that it allows
a direct application of a known key rate formula based on
unidirectional classical error correction and privacy amplification
\cite{kraus05a,renner05a,kraus07a}. Proposition~\ref{prop}, however, holds
independently of whether we restrict ourselves to unidirectional or
bidirectional classical communication protocols in the post-processing
stage.  

\begin{proposition}\label{prop} Whenever Alice and Bob use photon
  number diagonal measurement devices, then the secret key rate in the
  following two scenarios is the same:  
  \begin{enumerate}
  \item Eve distributes pure quantum states $\ket{\Psi}_{ \rm ABE}$
  from a given set $\mathcal{P}$ which contains the purifications of the
  states $\rho_{ \rm AB}=\tr_{\rm E}( \ket{\Psi}_{\rm ABE}\bra{\Psi})$ that are
  compatible with the observed measurement data.  
  \item Alice, Bob and Eve share tripartite states
  $\rho_{\rm ABE}=\tr_{\rm R}(\ket{\Phi}_{\rm ABER}\bra{\Phi}) \in
  \mathcal{S}$ which originate from a four-party state of the form  
  \begin{equation}
    \label{eq:shieldstate}
    \ket{\Phi}_{ \rm ABER}=\sum_{n,m=0}^\infty \sqrt{p_{nm}}
    \ket{\phi_{nm}}_{\rm ABE} \ket{n,m}_{\rm R},   
  \end{equation}
  where $\ket{\phi_{nm}}_{\rm ABE}$ represents a state that
  contains $n$ photons in the mode destined to Alice and $m$ photons in the
  mode for Bob, and where $\ket{n,m}_{\rm R}$ denotes an inaccessible shield
  system that records the photon number information of Alice and Bob's
  signals. The set of possible tripartite states $\mathcal{S}$ contains all
  such states for which the bipartite (photon-number diagonal) states
  $\rho_{\rm AB} = \tr_{\rm ER}(\ket{\Phi}_{\rm ABER}\bra{\Phi})$ are
  compatible with the observations.  
  \end{enumerate}
\end{proposition}

\begin{proof}

We show that for any state $\ket{\Psi}_{\rm ABE}\in\mathcal{P}$ chosen from
the first scenario, there is a particular three party state $\rho_{ABE} \in
\mathcal{S}$ from the second case such that Eve's position, once Alice and Bob
have performed their measurements, is completely equivalent. The reverse
direction of this statement holds trivially since $\mathcal{S} \subset
\mathcal{P}$. Note that Eve's eavesdropping capabilities are
completely determined by the collection of conditional states
$\rho_{\rm E}^{ij}$ and their corresponding probabilities $p_{ij}$,
both defined via the relation   
\begin{equation}
  \label{eq:Evesituation}
  p_{ij} \rho_{\rm E}^{ij} = \tr_{\rm AB} \left( F^{\rm A}_i \otimes F^{\rm
  B}_j \sigma_{\rm ABE} \right), 
\end{equation}
when Alice, Bob, and Eve 
share a state $\sigma_{\rm ABE}$. Let us start with the first case,
where $\sigma_{\rm ABE}=\ket{\Psi}_{ \rm ABE}\bra{\Psi}$ with
$\ket{\Psi}_{ \rm ABE}\in\mathcal{P}$. Using Eq.~\ref{eq:Evesituation}
and Eq.~\ref{mnw} we arrive at  
\begin{eqnarray}
  \nonumber
  p_{ij} \rho_{\rm E}^{ij} &=& \sum_{n,m} \tr_{\rm AB} \left( \Pi^{\rm A}_n
  F^{\rm A}_i \Pi^{\rm A}_n \otimes \Pi^{\rm B}_m F^{\rm B}_j \Pi^{\rm B}_m \;
  \ket{\Psi}_{\rm ABE}\bra{\Psi} \right) \\  
  \label{eq:situation1}
  &=& \sum_{n,m} p_{nm} \tr_{\rm AB} \left( F^{\rm A}_i \otimes F^{\rm B}_j
  \ket{\phi_{nm}}_{\rm ABE}\bra{\phi_{nm}} \right). 
\end{eqnarray}
In the second line we define $\Pi^{\rm A}_n \otimes \Pi^{\rm B}_m
\ket{\Psi}_{\rm ABE} = \sqrt{p_{nm}} \ket{\phi_{nm}}_{\rm ABE}$. To
compare it with the second scenario we select $\ket{\Phi}_{\rm ABER}$
arising from the state $\ket{\Psi}_{\rm ABE}$ via a coherent photon
number measurement. Its outcome is stored in the additional register
system ${\rm R}$ and the state is given by 
\begin{equation}
  \ket{\Phi}_{\rm ABER} = \sum_{n,m} \Pi^{\rm A}_n \otimes \Pi^{\rm B}_m
  \ket{\Psi}_{\rm ABE} \ket{0}_{\rm R} = \sum_{n,m} \sqrt{p_{nm}}
  \ket{\phi_{nm}}_{\rm ABE} \ket{n,m}_{\rm R}. 
\end{equation}
Using $\sigma_{\rm ABE} = \tr_{\rm R} \left( \ket{\Phi}_{\rm ABER}
\bra{\Phi}\right)$ in Eq.~\ref{eq:Evesituation} directly delivers the same
result as Eq.~\ref{eq:situation1}. This finally proves the proposition.   
\end{proof}

Next we focus on a security proof that only requires direct classical
communication in the reconciliation part of the protocol. More
precisely, we apply the secret key rate formula derived in the recent
security proof presented in Refs.~\cite{kraus05a,renner05a,kraus07a}. It
relies on Alice, Bob, and Eve sharing signal states of the form given
by Eq.~\ref{eq:shieldstate}. Once the legitimate users have measured
their part of $\rho_{\rm ABE}=\tr_{\rm R}(\ket{\Phi}_{\rm ABER}\bra{\Phi})$, they
only have access to their classical outcomes which are stored in
registers $\rm X$ and $\rm Y$ respectively. On the contrary, Eve still
has at her disposal a quantum state. This scenario is described by the
so-called ccq state $\rho_{\rm XYE}=\mathcal{M}(\rho_{\rm ABE})$ that
results from the map   
\begin{equation}
  \label{eq:ccq}
  \rho_{\rm ABE}  \mapsto \rho_{\rm{XYE}}=\mathcal{M}(\rho_{\rm ABE})=
  \sum_{i,j} p_{ij} \ket{i,j}_{\rm XY}\bra{i,j} \otimes \rho_{\rm E}^{ij},
\end{equation}
where the probabilities $p_{ij}$ and the conditional states
$\rho_E^{i,j}$ are defined like in Eq.~\ref{eq:Evesituation} by
setting $\sigma_{\rm ABE}=\rho_{\rm ABE}$. According to
Refs.~\cite{kraus07a}, the secret key rate,
that we shall denote as $R$, satisfies 
\begin{equation}
  \label{eq:rate1}
  R \geq \inf_{\rho_{\rm ABE} \in \mathcal{S}} \sum_{n,m=0}^{\infty} g_{nm}
  S(X \vert E,n,m)- g H(X \vert Y).  
\end{equation}
The infimum runs over all possible tripartite states $\rho_{\rm ABE}$ that 
belong to the class $\mathcal{S}$ defined in Proposition \ref{prop}.
Here $H(X \vert Y)$ stands for the conditional Shannon entropy of Alice's
random variable $X$ conditioned on Bob's random variable $Y$. This 
part accounts for the error correction step of the protocol
and it is independent of the chosen tripartite state $\rho_{\rm ABE} \in
\mathcal{S}$. In order to compute the conditional von Neumann
entropies $S(X \vert E,n,m)$ describing Eve's information about
Alice's raw key, we first calculate the ccq states for the definite
photon number states $\ket{\phi_{nm}}_{\rm  ABE}$ as they appear in
the decomposition $\rho_{\rm ABE}=\sum_{n,m} p_{nm}
\ket{\phi_{nm}}_{\rm ABE}\bra{\phi_{nm}}\in\mathcal{S}$. Let us denote
these conditional ccq states as $\rho_{\rm XYE}^{nm} =
\mathcal{M}(\ket{\phi_{nm}}_{\rm ABE} \bra{\phi_{nm}})$. From the
definition of the conditional entropy we obtain $S(X \vert E, n,m) =
S( \rho_{\rm XE}^{nm}) - S( \rho_{\rm E}^{nm})$, where $S(\rho)$
denotes the von Neumann entropy of a generic quantum state $\rho$. The
remaining parameters that appear in Eq.~\ref{eq:rate1} are the overall
gain $g$ and the individual gains $g_{nm}$, \ie, the probability that
Alice and Bob obtain an overall conclusive result when $n$ and $m$
photons are detected on each side respectively. These parameters can
be written as    
\begin{eqnarray}
  g_{nm} &=& p_{nm} Y_{nm}, \\ 
  g &=&\sum_{n,m} g_{nm},
\end{eqnarray}
with the conditional yields $Y_{nm}$ defined as the probability that both
parties obtain a conclusive outcome conditioned on the fact that they
received a state $\tr_{\rm E}(\ket{\phi_{nm}}_{\rm ABE}\bra{\phi_{nm}})$.

The detector decoy idea does not imply any change in the underlying
security proof. In fact, one could even improve the 
lower bound on the secret key rate
formula given in Eq.~\ref{eq:rate1} by including local randomization
steps \cite{kraus05a,renner05a,renes07a,smith08a} or by allowing
several rounds of classical bidirectional communication
\cite{kraus07a}. The main advantage of the detector decoy method is
that it allows Alice and Bob to acquire more information about the
class $\mathcal{S}$ over which they have to perform the
optimization. As explained in the next subsection, one can in
principle obtain the \emph{full statistics} that PNR detectors could
give. Note, however, that when Alice and Bob use PNR detectors they
also have single shot resolution. Still, to have access to the
statistics of the arriving signals allows the legitimate users to gain
more knowledge about Eve's information on the raw key. Thus a smaller
amount of privacy amplification is needed, and consequently one
obtains more secret key.  

One can further simplify the lower bound on the secret key rate formula
such that only a few parameters need to be estimated,
cf. Ref.~\cite{kraus07a}. The conditional entropies satisfy $H(X
\vert n,m) \geq S( X \vert E,n,m) \geq 0$ for all photon numbers $n$
and $m$. This means that the secret key rate $R$ can always be lower
bounded by restricting the sum in Eq.~\ref{eq:rate1} to any of its
items. For instance, one can select the single photon and the vacuum
contributions only, and obtains   
\begin{eqnarray}
  \nonumber
  R \geq \inf_{\rho_{\rm ABE}\in\mathcal{S}} \sum_{m=0}^\infty
     && g_{0m}  S(X \vert n=0,m) + g_{11} S(X\vert E,n=1,m=1) \\
  &&
  \label{eq:rate_2}
  -g H(X \vert Y).   
\end{eqnarray}

So far we have not considered any explicit QKD scheme yet. In the
following we restrict ourselves to the BB84 and the 6-state protocols,
since they allow us to express both entropies by means of quantities
that are directly observable. Moreover, and for simplicity, let us
assume that the sifted key is only composed by those events for which
Alice and Bob have used their normal detection device, \ie, all
possible decoy outcomes are considered as inconclusive results and
they are only used to estimate the class $\mathcal{S}$. Similarly, all
the no click outcomes and all detection events where Alice and Bob
employed different basis choices are discarded as well. In the case of
double clicks two options are possible: Either they are discarded as
well, or one assigns at random one of the two conclusive outcomes
``0'' or ``1'' \cite{nor_99}. As a result, Alice and Bob are left with
binary values whenever they consider an outcome pair as conclusive. As
shown in Ref.~\cite{kraus05a,renner05a}, both parties can randomly flip their
bit values together, which results in an overall symmetric error rate
that gives $H(X \vert Y)=h_2(Q)$, where $h_2$ denotes the binary
entropy. Any conclusive result on Alice's side that originates from a
vacuum input contains no information for the eavesdropper \cite{lo05b}. If we
assume that these events are  completely unbiased we obtain $S(X \vert
E,n=0,m)=1$ for all $m$. This means that Alice and Bob do not need to
perform any privacy amplification on all these outcomes, but note 
as well that they do not provide any key information because of the
error correction part in the formula. The total \emph{vacuum gain} is given by
$g_{0}=\sum_{m} g_{0m}$. The conditional von Neumann entropy from the
single photon contribution can always be lower bounded by the
completely symmetric case, which gives $S(X \vert E, n=1,m=1) \geq
f(Q_{11})$, where $f$ denotes a convex function that depends on the
chosen protocol, and $Q_{11}$ represents the conditional single photon
QBER that one observes with \emph{perfect} detectors. For the two
considered protocols this function $f$ takes the form 
\footnote{Let us mention two important points here. Because of convexity of
  both functions one could alternatively use the actual, single photon
  QBER as an argument of the lower bound functions. This situation
  corresponds to the case in which one assumes the uncalibrated device
  scenario for the evaluation of the privacy amplification
  part. However, if one takes into account any imperfections
  from the actual detection device, then one could even enhance the
  actual lower bound $f$. For example, if one considers a dark count
  model that randomly flips the bit value on Alice side (dark counts
  produce double clicks which are randomly assigned afterwards) this
  actually reduces Eve's information on the raw key and hence the
  privacy amplification part \cite{kraus05a,renner05a}. Nevertheless we shall
  ignore this effect in our discussion.}   
\begin{equation}
  \label{eq:privacy}
  f(x)= \left\{ \begin{array}{ll} 1-h_2(x) & \textrm{BB84,} \\
      1+h_2(x)-h_2(\frac{3x}{2})-\frac{3x}{2} \log_2(3) &
      \textrm{6-state.} \end{array} 
  \right.  
\end{equation}
Let $g_{11}^{\rm min}$ and $g_{0}^{\rm min}$ denote lower bounds on the
single photon and vacuum gain respectively, while $Q_{11}^{\rm
  max}$ represents an upper bound to the maximal attainable value of
the single photon QBER, of all states compatible with the class
$\mathcal{S}$. With this notation, the secret key rate satisfies   
\begin{equation}\label{rate_mc}
  R \geq  g^{\rm min}_{0}+g^{\rm min}_{11} f(Q^{\rm max}_{11}) -   g h_2(Q),
\end{equation}
with the distinction between the BB84 and the 6-state protocol being
only in the function $f$ given by Eq.~\ref{eq:privacy}. Note that the gains
$g_{11}, g_0$ and $g$ depend on the choice of which outcomes are
considered as conclusive. This decision includes as well the overall sifting
effect. Let $q$ denote the probability that both parties use their
normal detection device and they measure in the same basis.  Then, using an
asymmetric basis choice in the setup in combination with a very
rarely switching to the decoy measurement, this overall sifting factor
$q$ can be made arbitrary close to unity \cite{efficientBB84} and thus
we can drop it in the evaluation section.

\subsection{Detector decoy estimation}\label{decoy_est}

In this section we apply the detector decoy method to the active
measurement setup as it is used in the usual BB84 or the 6-state
protocol, and we show how Alice and Bob can estimate the
essential quantities to evaluate Eq.~\ref{rate_mc}. That is, the vacuum gain
$g_0$, the single photon gain $g_{11}$ and the conditional quantum bit
error rate $Q_{11}$ from perfect detectors.  

The discussion starts with the typical model of an imperfect
threshold detector which shows some noise in the form of dark counts
as given by Eq.~\ref{eq:inefficient} with $\eta=1$. Furthermore we
require that Alice's and Bob's detectors have equal (and constant) detector
inefficiency; otherwise this opens the possibility for powerful new
eavesdropping attacks \cite{makarov06a,qi07a} and other techniques
have to be applied \cite{fung08a,lydersen08a}. Under this assumption it is a
common technique to include the inefficiency of the detectors into the
action of the quantum channel, and one performs the analysis with a
threshold detector model of perfect efficiency. If one
can prove security without knowing the exact detector efficiency,
then one automatically also shows security with this particular extra
knowledge. Suppose that both threshold detectors on each side have
equal dark count probabilities. The POVM elements for the active
measurement choice $\beta$ are given by \cite{nor_99} 
\begin{eqnarray}
  \label{eq:F0}
  F_{\rm 0}&=&(1-\epsilon)  \sum_{n= 1}^\infty
  \ket{n,0}_\beta\bra{n,0} +   \epsilon (1-\epsilon)
  \ket{\rm{0,0}}\bra{\rm{0,0}}, \\ 
  \label{eq:F1}
  F_{\rm 1}&=& (1-\epsilon) \sum_{n=1}^\infty \ket{0,n}_\beta\bra{0,n}
  + \epsilon (1-\epsilon) \ket{\rm{0,0}}\bra{\rm{0,0}}, 
\end{eqnarray}
together with $F_{\rm vac}=(1-\epsilon)^2 \ket{\rm{0,0}}\bra{\rm{0,0}}$ and 
$F_{\rm D}=\mathbbm{1} - F_{\rm vac}- F_{\rm 0}- F_{\rm 1}$, with
$F_{\rm D}$ denoting the operator associated with double click events.  
Here $\ket{n,0}_\beta$ and $\ket{0,n}_\beta$ refer to the
corresponding two-mode Fock state in the chosen polarization basis
$\beta$. Although we restrict ourselves to this particular model, the
analysis that follows can also be straightforwardly adapted for the
calibrated device scenario.     

Let us begin by analyzing the single photon gain $g_{11}$. Consider a 
simple setup where Alice and Bob insert only a single beam splitter in
front of their measurement devices. This beam splitter is of course
not assigned to the quantum channel. This scenario is illustrated in
Fig~\ref{fig:simpledecoy}, where the transmittance of Alice and Bob's
beam splitter is denoted as $\eta_{\rm A}$ and $\eta_{\rm B}$ respectively.   
\begin{figure}[ht]
  \centering
  \resizebox{100mm}{!}{\includegraphics{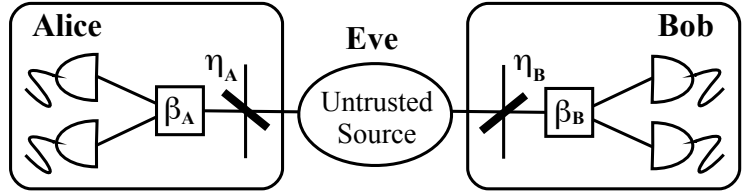}}
  \caption{Schematic diagram of the first detector decoy setup
    considered. Alice and Bob place a single beam splitter, with
    transmittance $\eta_A$ and $\eta_B$ respectively, in front of
    their detection device.} 
  \label{fig:simpledecoy}
\end{figure}
This setup has the advantage that both legitimate users only need to
collect the count rates for one variable beam splitter per site, but
it still enables Alice and Bob to obtain the overall photon number
distribution $p_{nm}$ of the incoming signals. With this information
they can directly compute the individual gain of the single photon
contribution $g_{11}$ for the considered scenarios. For this detection
device the overall ``no click'' operator on Alice's side becomes 
\begin{equation} 
 F^{\rm A}_{\rm vac}(\eta_{\rm A})= (1-\epsilon)^2 \sum_{n=0}^\infty
 (1-\eta_{\rm A})^n \Pi^{\rm A}_n, 
\end{equation} 
where the projector onto the $n$-photon subspace is given by
\begin{equation}
  \Pi^{\rm A}_n = \sum_{k=0}^n \ket{k,n-k}_\beta \bra{k,n-k}.
\end{equation}
Note that this operator is independent of the chosen measurement basis
$\vec \beta$, and hence we omit this label in the following whenever
it is redundant. A similar expression holds for Bob's measurement operator
$F^{\rm B}_{\rm vac}(\eta_{\rm B})$. Suppose that both parties receive
now the generic input state $\rho_{\rm AB}$, then they observe a no
click outcome with probability  
\begin{equation}
  \label{eq:decoy1simple}
  p^{\rm AB}_{\rm vac}(\eta_{\rm A},\eta_{\rm B}) = (1-\epsilon)^4
  \sum_{n,m=0}^\infty (1- \eta_{\rm A})^n (1 - \eta_{\rm B})^m p_{nm},
\end{equation}
where the photon number distribution $p_{nm}$ is given by
$p_{nm}=\tr_{\rm AB}( \Pi_n^{\rm  A} \otimes \Pi_m^{\rm B} \rho_{\rm
  AB})$. Now, if Alice and Bob vary the transmittance of their
inserted beam splitters, they can generate a whole set of linear
equations similar to those given by Eq.~\ref{eq:decoy1simple}, in
which the photon number distribution $p_{nm}$ appears as the open
parameter. With this set of equations the whole distribution becomes
accessible to Alice and Bob, however if they are only interested in
the single photon probability $p_{11}$ then the three different
settings of Proposition~\ref{prop_n} are already enough. Next, let us
compute the single photon gain $g_{11}$ for two different
scenarios. Whenever Alice and Bob randomly assign bit values to their
double click outcomes any single photon state will necessarily produce
a conclusive outcome. In the second case, we consider that only single
clicks contribute to the raw key rate. Here the individual gain is
slightly lower than in the first situation since a single photon can
trigger a double click event because of dark counts. The two different
individual gains are given, respectively, by   
\begin{eqnarray}
  g_{11, \rm d}&=&p_{11}, \\
  g_{11, \rm s}&=&(1-\epsilon)^2 p_{11},
\end{eqnarray}
where the subscripts ``d'' (with double clicks) and ``s'' (single
clicks only) label the two different cases. Let us mention that the
idea of obtaining the impinging photon number statistics with only one
variable beam splitter can as well be applied to other
photon number preserving linear networks, since the probability to
obtain an overall no click outcome in the all the threshold detectors after
such a network can always be calculated by replacing the whole network
by only one such threshold detector.  

The vacuum gain $g_0$ represents a direct observable quantity even
without the detector decoy method. It only relies on the fact that one
can obtain the statistics of a perfect threshold detector from the
observed data of a detector which has dark counts \cite{moroder06a}. 
Suppose that Bob considers a specific measurement outcome $k$ which he
can perfectly distinguish with his measurement device, and the
corresponding POVM element is denoted by $F_k^{\rm B}$. Then, the
probability that Alice registers no click at all while Bob sees this
specific outcome is given by   
\begin{equation}
  \label{eq:p0kdark}
  p_{{\rm vac},k}^{\rm AB}= (1-\epsilon)^2 \tr_{\rm AB}
  \left(\ket{\rm 0,0}_{\rm A}\!\bra{\rm 0,0} \otimes F_k^{\rm B} \rho_{\rm
      AB}\right)=(1-\epsilon)^2 p_{0\rm k}.  
\end{equation}
Since both parties can have access to the dark count probability
$\epsilon$ they can directly use Eq.~\ref{eq:p0kdark} to compute the
probability $p_{0,\rm k}$ from perfect threshold detectors. Using this
value directly allows to infer the vacuum gain for the two different
scenarios as  
\begin{eqnarray}
  g_{0, \rm d}&=&[1-(1-\varepsilon)^2] p_{0,\rm d}, \\
  g_{0, \rm s}&=&2 \epsilon(1-\epsilon) p_{0,\rm s}.
\end{eqnarray}

The resolved single photon QBER $Q_{11}$ is \emph{inaccessible}
with the simple detector decoy setup presented above. 
We can consider a more complicated setup in which a variable beam
splitter is place in front of each threshold detector. This scenario
is depicted in Fig.~\ref{fig:harddecoy}. Now Alice and Bob can adjust the
transmittance of their two beam splitters $\vec \eta_{\rm
  A}=(\eta_{{\rm A},1},\eta_{{\rm A},2})$ and $\vec \eta_{\rm
  B}=(\eta_{{\rm B},1}, \eta_{{\rm B},2})$ respectively. Although,
from a practical point of view, this scenario is less attractive than
the previous one---it requires a more complicated statistical
analysis---it is interesting on a conceptual level since it can
provide Alice and Bob with the same statistics like PNR
detectors. Obviously all the results from the simple setup apply
if one selects $\eta_{{\rm A},1}=\eta_{{\rm A},2}=\eta_{\rm A}$ and
similar for Bob's side.   
\begin{figure}[ht]
  \centering
  \resizebox{100mm}{!}{\includegraphics{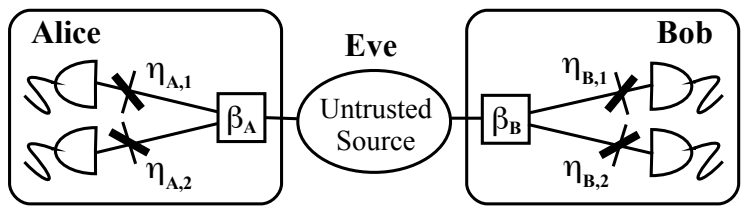}}
  \caption{Schematic diagram of the second scenario analyzed, where
    Alice and Bob place a variable beam splitter in front of every
    threshold detector.} 
  \label{fig:harddecoy}
\end{figure}
Now the POVM element for the overall no click outcome on Alice's side
is given by   
\begin{equation}
  F^{\rm A}_{\rm vac,\beta_A} (\vec \eta_{\rm A})= (1-\epsilon)^2
  \sum_{k,l=0}^\infty  \bar \eta_{\rm A,1}^k \bar \eta_{\rm A,2}^m
  \ket{k,l}_{\beta_A}\!\bra{k,l}, 
\end{equation}
where we use the abbreviation $\bar \eta = 1- \eta$. A similar
expression can be obtained
for $F^{\rm B}_{\rm vac, \beta_B} (\vec \eta_{{\rm B}})$. In
contrast to the first measurement device analyzed, now the no click
outcome depends on the chosen polarization basis $\vec
\beta=(\beta_{\rm A}, \beta_{\rm B})$. The probability of the combined
``no click'' outcome in Alice's and Bob's side can be written as  
\begin{equation}
  \label{eq:decoy2hard}
  p_{{\rm vac},{\vec \beta}}^{\rm AB}(\vec \eta_{\rm A}, \vec
  \eta_{\rm B})= (1-\epsilon)^4  \sum_{k,l,r,s=0}^{\infty}
  \bar\eta_{{\rm A},1}^k \bar\eta_{{\rm A},2}^l \bar\eta_{{\rm B},1}^r
  \bar\eta_{{\rm B},2}^s \;q_{\vec \beta}(k,l;r,s), 
\end{equation}
where the probabilities $q_{\vec \beta}(k,l;r,s)$ have the form
\begin{equation}
  q_{\vec \beta}(k,l;r,s) = \tr_{\rm AB} \left( \ket{k,l}_{\beta_{\rm
        A}}\bra{k,l} \otimes \ket{r,s}_{\beta_{\rm B}}\bra{r,s} \;
    \rho_{\rm AB}\right).  
\end{equation}
These probabilities coincide with the ones provided by PNR
detectors. Using again the idea of different settings for the
transmittance of the adjustable beam splitters one can generate more
linear equations of the form given by Eq.~\ref{eq:decoy2hard}. 
Consequently, the photon number resolved statistics $q_{\vec \beta}(k,l,r,s)$ 
become accessible to Alice and Bob. With this resolved distribution at
hand it is straightforward to compute the single photon QBER $Q_{\vec \beta,11}$
that Alice and Bob would observe using perfect detectors. It is
determined by    
\begin{equation}
  p_{11} Q_{\vec \beta,11}=q_{\vec \beta}(1,0;0,1)+q_{\vec
    \beta}(0,1;1,0).
\end{equation}
Note that one has to use the symmetrized single photon QBER in the lower
bound formula given by Eq.~\ref{rate_mc}.

\subsection{Evaluation}\label{evalua}

In this part we evaluate the lower bound on the secret key rate for
the different decoy detection schemes presented in the last
subsection. Additionally, we compare it with a security proof that
relies on the validity of the squash model; for a different comparison
between the squash model and an alternative estimation procedure not
based in this last paradigm see also Ref.~\cite{koashinew}. We assume
that all relevant parameters that appear in the lower bound formula
can be estimated precisely, \ie, we ignore any statistical effect of
an estimation procedure that uses only a few number of decoy settings. 

Suppose that the observed data originate from a pumped type-II down conversion
source. The states emitted by this type of source can be written as
\cite{PDC_source} 
\begin{equation}
  \ket{\Psi_{\rm source}}_{\rm AB}=\sum_{n=0}^{\infty} \sqrt{p_n}
    \ket{\Phi_n}_{\rm AB},  
\end{equation}
where the probability distribution $p_n$ is given by
\begin{equation}
  p_n=\frac{(n+1)\lambda^n}{(1+\lambda)^{n+2}}.
\end{equation}
The parameter $\lambda$ is related with the pump amplitude of the
laser and determines the mean photon pair number per pulse
as $\mu=2 \lambda$. Each signal state $\ket{\Phi_n}_{\rm AB}$
contains exactly $2n$ photons; $n$ of them travel to
Alice and the other $n$ to Bob. These states are of the form
\begin{equation}\label{mar}
  \ket{\Phi_n}_{\rm AB}=\sum_{m=0}^n \frac{(-1)^m}{\sqrt{n+1}}
  \ket{n-m,m}_{\rm A}\ket{m,n-m}_{\rm B},  
\end{equation}
where we have used the standard basis on each
side, {\it i.e.}, $\beta_{\rm A}=\beta_{\rm B}=z$. When 
$n=1$, the signal state in Eq.~\ref{mar}
becomes the EPR state, which admits perfect anticorrelations in all
directions\footnote{If Alice and Bob employ the measurement devices
  from Eqs.~\ref{eq:F0},~\ref{eq:F1}, then they would always observe
  anticorrelated outcomes. Hence, one of the parties has to
  interchange the observed data ``0'' $\leftrightarrow$ ``1''.}. 
When $n \geq 2$, the states $\ket{\Phi_n}_{\rm AB}$ represent
$W$-states. That is, even if Alice and Bob measure them along the same
direction they might observe double clicks. In fact, the biggest
contribution in the observed QBER stems from the multiphoton
pairs. For instance, if the signal $\ket{\Phi_{n=2}}$ loses only one
photon in the channel, then the error rate of the resulting state
(although still entangled) is
already about\footnote{Note, however, that there are different, more
  complicated measurement techniques that can be more robust against
  particle loss from a PDC source \cite{norbertcomment}.}
$16.6\%$. This QBER is above the threshold error rate allowed by the
one-way security proof employed in the previous section, even assuming
a single qubit realization. This means, in particular, that the
expected average mean photon number $\lambda$ which optimizes the
secret key rate in the long distance limit is quite low, and one does
not expect a security proof which enables to drive the source with a
much higher mean photon number.    

To generate the observed data of an experiment that uses this kind of
source we employ the following procedure: Since the loss is the
predominant factor in the error rate and in the overall gain, we
assume that the state emitted by the source passes first through a
lossy, but otherwise error-free channel. Such a channel is
characterized by the loss coefficient $\alpha$ and the total distance
$l$. We include as well in the channel the effect of the detector
efficiency $\eta_{\rm det}$ of the measurement device. 
Hence, the overall transmission in the optical line towards Alice becomes  
\begin{equation}
  \eta_{\rm A}= \eta_{\rm det} 10^{-\frac{\alpha l}{10}} =
  10^{-\frac{\db_{\rm A}}{10}},
\end{equation}
and determines the overall loss coefficient $\db_{\rm A}$. A similar
relation holds also for the channel towards Bob. The total loss
between both parties is characterized by $\db_{\rm tot}=\db_{\rm A}+\db_{\rm B}$.
After the lossy channel, we include the effect of the misalignment and
the dark counts of the detectors in the observed data. The
misalignment varies slightly the polarization of the incoming light
field. This effect changes over time and it is assumed to be
uncontrollable. Averaging it results in an action similar to that of a
depolarizing channel. Specifically, we consider the following misalignment
model: Every time a single photon arrives at the detection device it
triggers the correct detector with probability $(1-e)$, while with
probability $e$ it changes its polarization and triggers the wrong
detector. When more photons enter the detection apparatus this effect
is assumed to occur independently for every single photon and hence it can
also change the probability to observe a double click. 
To conclude, we assume as well that Alice's and Bob's detectors suffer
from dark counts as described in the previous subsection. Dark counts
are typically the crucial parameter that limit the distance of a QKD scheme. 


Next, we compare the different lower bounds on the secret key rate for
the various scenarios considered, which again are distinguished by
the subscripts ``s'' and ``d'' depending on the double click
choice. In the simple detector decoy setup the resolved
error rate is not directly accessible. Still, one can upper bound it
via a worst case assumption. That is, we consider that \emph{all}
errors originate from the single photons only. In the single click
case this upper bound, denoted as $\overline Q_{11,\rm s}$, is 
given by\footnote{From this estimation one could even try to calculate
  out the dark count probability of the detectors. However, we shall
  ignore this effect here.}    
\begin{equation}
   \overline Q_{11,\rm s} = \frac{1}{g_{11,s}} \left( Q_{\rm s} g_{\rm
       s} - g_0 \frac{1}{2}\right) \geq Q_{11}. 
\end{equation}
A similar relation gives the upper bound $\overline Q_{11,\rm d}$ for
the double click case. On the other hand, in the detector decoy scheme
which has two variable beam splitters in each side the conditional
quantum bit error rate $Q_{11}$ is equal to the hypothetical QBER
arising with perfect detectors, independently of the chosen
scenario. The different lower bounds are illustrated in
Tab.~\ref{tab:rates}.    
\begin{table}[ht]
  \centering
  \begin{tabular}{cc}
    \hline\hline
    Scenario \T \B & Lower bound \\
    \hline\hline
    Updated Squash \T \B & $\tilde g_0 + \tilde g_{11}[1-h_2(\tilde
    Q_{11})] - g_{\rm d} h_2(Q_{\rm d})$  \\
    \hline
    Double \T \B & $g_{0,\rm d} + g_{11,\rm d} f(Q_{11}) - g_{\rm d}
    h_2(Q_{\rm d})$ \\    
    Double + Bound \T \B & $g_{0,\rm d} + g_{11,\rm d} f(\overline
    Q_{11,\rm d}) - g_{\rm d} h_2(Q_{\rm d})$ \\    
    \hline
    Single \T \B & $g_{0,\rm s} + g_{11,\rm s} f(Q_{11}) - g_{\rm s}
    h_2(Q_{\rm s})$ \\    
    Single + Bound \T \B & $g_{0,\rm s} + g_{11,\rm s} f(\overline
    Q_{11,\rm s}) - g_{\rm s} h_2(Q_{\rm s})$ \\   
    \hline 
    PNR \T \B & $p_{11} [ f(Q_{11}) - h_2(Q_{11})]$ \\
    \hline \hline
  \end{tabular}
  \caption{Different lower bounds on the secret key rate for the
    various scenarios considered with active basis choice measurements.} 
  \label{tab:rates}
\end{table}
This table also includes the case where both parties employ
\emph{perfect} PNR detectors. This type of detectors allows them to
condition the error correction on the photon number observed. This
way, the term $H(X\vert Y)$ which appears in Eq.~\ref{eq:rate1} can
be changed by the conditional term $\sum g_{nm} H(X \vert Y,n,m)$. The
lower bound contained in Tab.~\ref{tab:rates} corresponds to the case where
single click events are the only conclusive outcomes.  
A fair comparison with a security proof based on the squash model is
only possible if one employs the result from Refs.~\cite{squash1,squash2}
to further extend the validity of the squash model to the situation where 
Alice's and Bob's detectors have dark counts; otherwise one loses the
vacuum gain. The exact target measurements are given by
Eqs.~\ref{eq:F0},~\ref{eq:F1} with $n=1$ and where every double click
is randomly assigned to one bit value. As a result, we obtain the
lower bound formula given in Tab.~\ref{tab:rates} (``Updated squash'')
in which $\tilde g_0$ and $\tilde g_{11}$ denote, respectively, the
squashed vacuum gain and the corresponding single photon gain, while
$\tilde Q_{11}$ stands for the conditional QBER on the squashed single
photons. Since the squash model does not exist for the active 6-state
protocol\footnote{All formulae in Tab.~\ref{tab:rates}
    are for active basis choice measurements. In particular in a
    passive basis choice selection the results of the squash model
    change. One could at least discard double or multiclick events
    between different basis outcomes, and moreover a squash model
    exists for the passive 6-state protocol.} 
\cite{squash2}, here one cannot choose the function $f$. 

For the simulation we consider an asymmetric distance scenario, since
\begin{figure}[h!h!t!]
  \centering
  \includegraphics[angle=-90,scale=0.40]{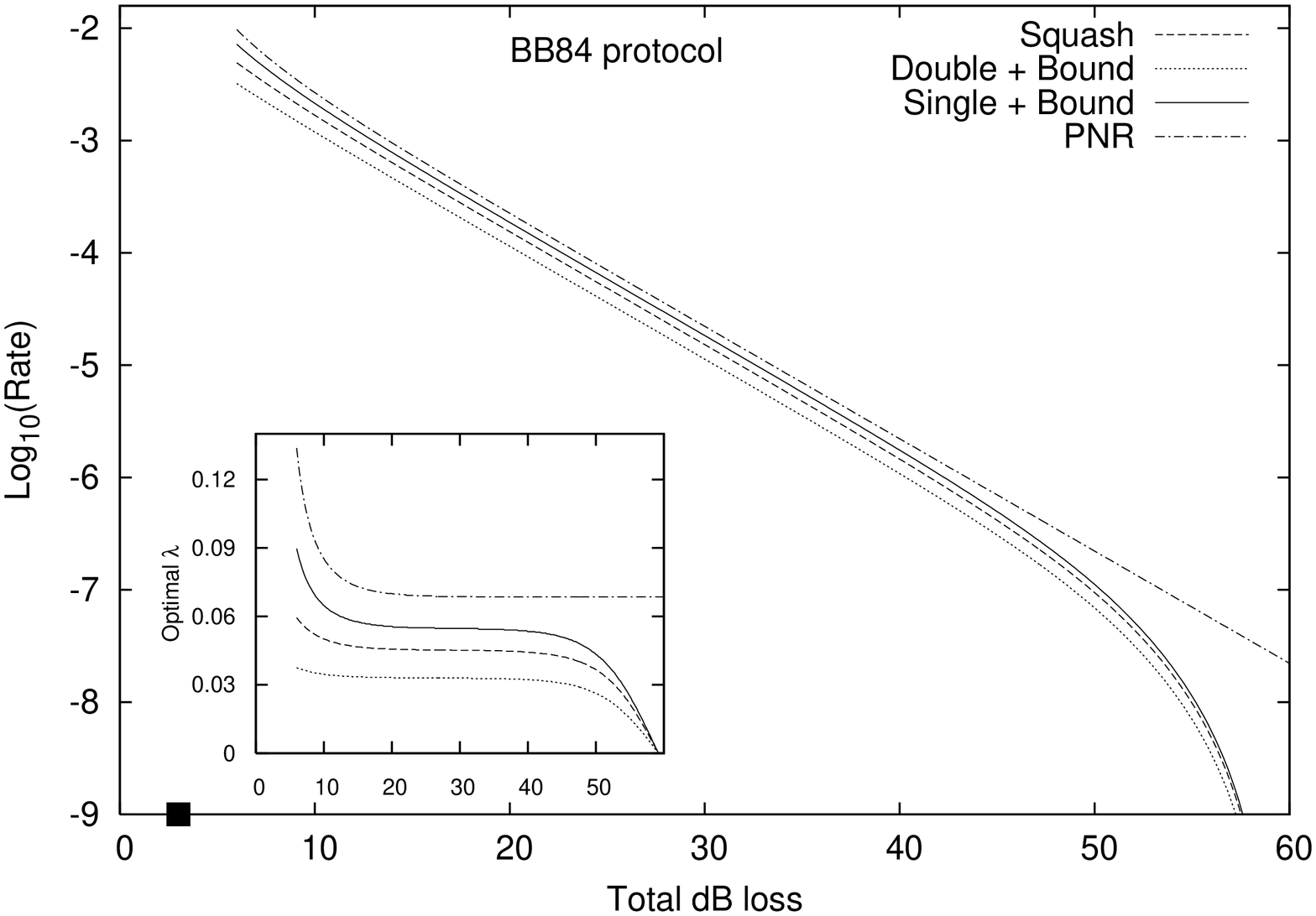}
  \includegraphics[angle=-90,scale=0.40]{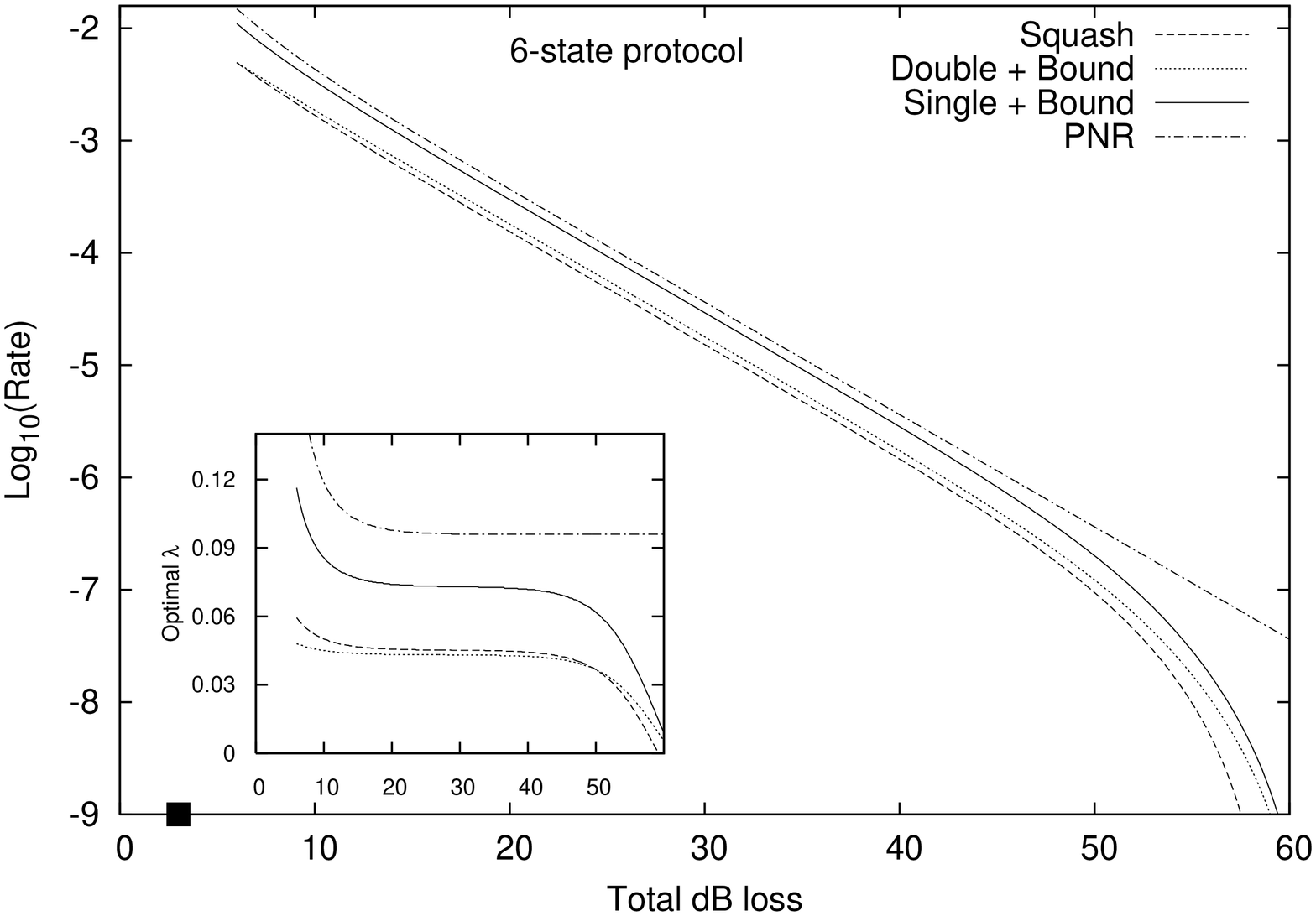}
  \caption{Different lower bounds on the secret key rate for the
    simple decoy detection setup shown in Fig.~\ref{fig:simpledecoy}
    with $\epsilon=10^{-6}$ and $e=0.03$. The inset figure shows the
    value for the optimized parameter $\lambda$ of the source.} 
  \label{fig:ratebb85}
\end{figure}
\begin{figure}[h!h!t]
  \centering
  \includegraphics[angle=-90,scale=0.40]{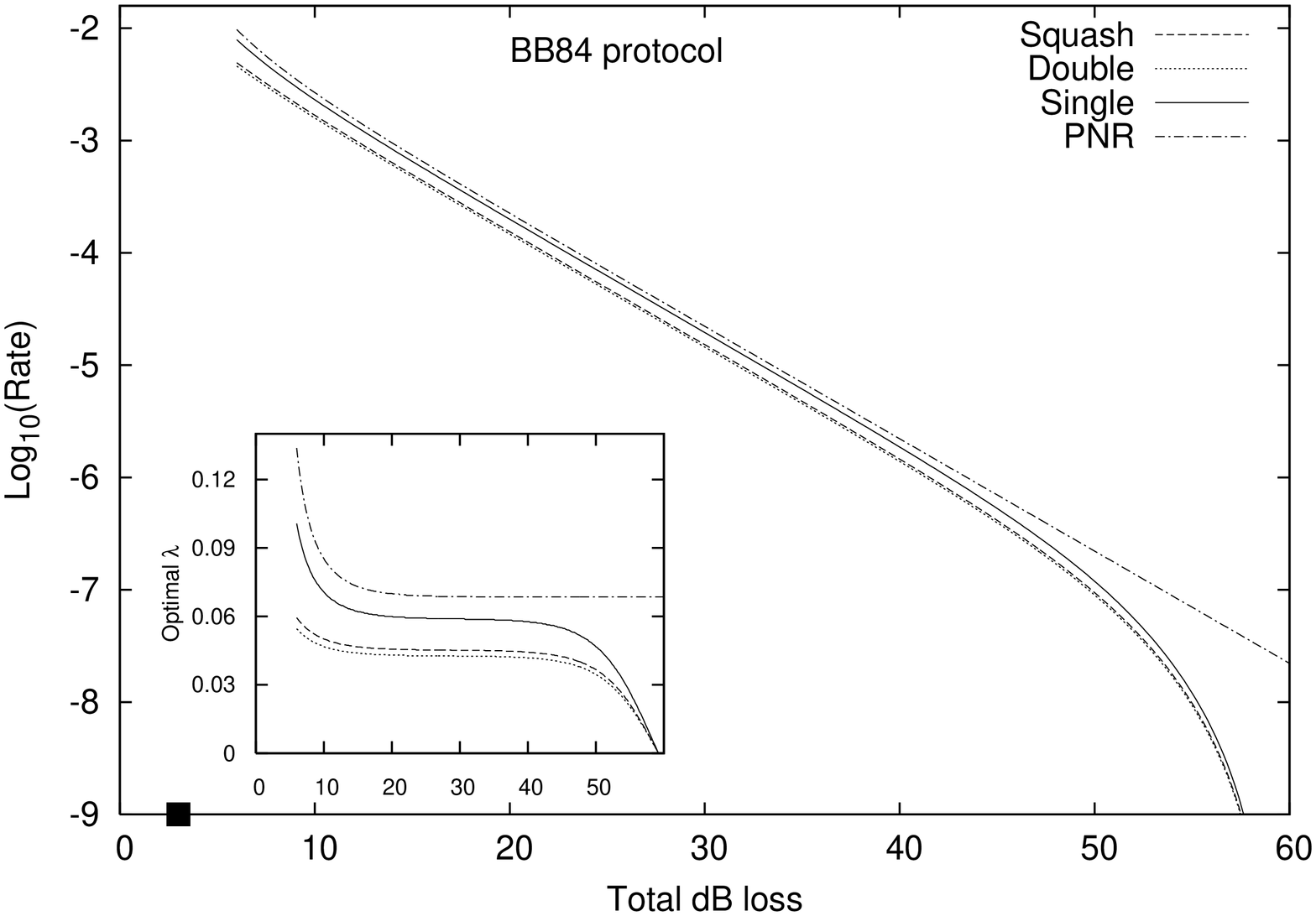}
  \includegraphics[angle=-90,scale=0.40]{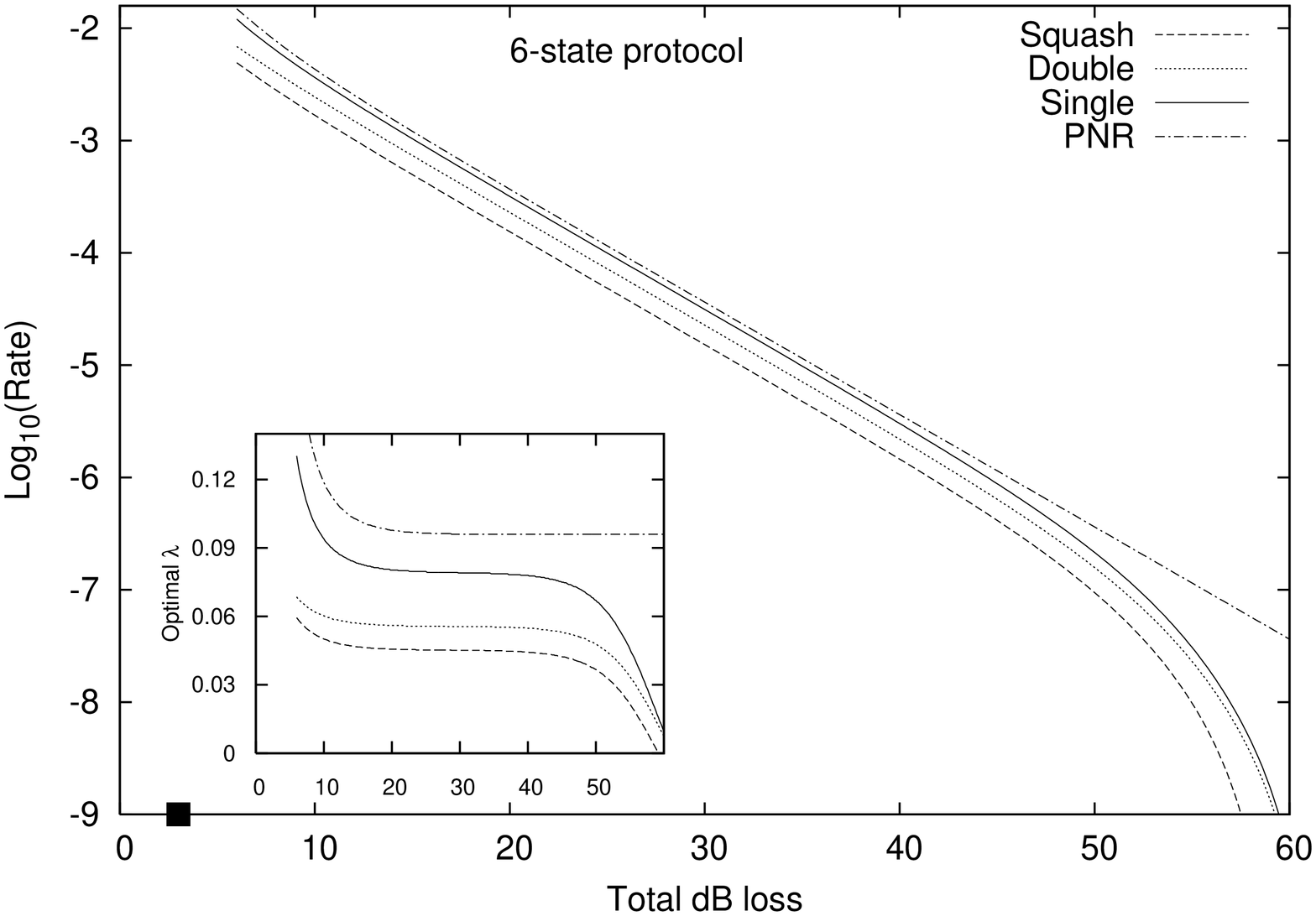}
  \caption{Different lower bounds on the secret key
    rate for the refined decoy detection setup shown in
    Fig.~\ref{fig:harddecoy} with $\epsilon=10^{-6}$ and
    $e=0.03$. The inset figure shows the value for the optimized
    parameter $\lambda$ of the source.}
  \label{fig:ratebb84}
\end{figure}
this situation optimizes the gain of the detector decoy idea over a
security proof that relies on the squash model. In particular, we
assume that Bob is much closer to the source than Alice. Such a
situation might appear often in a QKD network, where certain users can be
further away from the relay stations than others.  
The results for the BB84 and the 6-state protocol are shown, respectively, in
Fig.~\ref{fig:ratebb85} and in Fig.~\ref{fig:ratebb84}. The
first case always corresponds to the situation where Alice and Bob
place a single beam splitter in front of their detection device, while
the second case represents the scenario where the legitimate users
place a variable beam splitter in front of every threshold
detector. The position of the source is denoted by a black   
square and is kept constant at a $\db_{\rm B}=3$ loss distance, so we
only increase the distance towards Alice. For each lower bound we
perform an optimization over the free parameter $\lambda$ that
corresponds to the mean photon pair number. 

It is worth mentioning that the squash model delivers a higher lower
bound on the secret key rate than that corresponding to the detector
decoy method in the double click case. This seems surprising at first
since the detector decoy idea provides the exact knowledge of all
important single photon parameters. Note, however, that in the
discussion which leaded to the lower bound formula given by
Eq.~\ref{eq:rate_2} we restricted ourselves to only draw a secret key from
the single photon contribution. In contrast, the squash model
does not necessarily constrain the parties to obtain a secret key from
the single photon contribution only, but instead attempts to
even draw a secret key from the multiphoton events. In this sense, one
can consider the squash model as a ``calculation method'' that allows
to lower bound the amount of privacy amplification necessary for the
multiphoton events by an equal amount of privacy amplification
``calculated'' on a hypothetical single photon state. Hence using the
squash model directly lower bounds the key rate from
Eq.~\ref{eq:rate1}. On the contrary, the detector decoy idea provides 
a slightly higher secret key rate than the squash model when Alice and
Bob discard their double click events. Note that this action is not
possible with the squash model assumption. For the asymmetric distance
scenario, this fact allows the two parties to drive the source with a
slightly higher mean photon number, since the double clicks that occur
frequently on the side closer to the source can be discarded from the
error rate. See inset plots of the optimized mean photon number in
Figs.~\ref{fig:ratebb85},~\ref{fig:ratebb84}. In the
squash model one has to keep the double click rate low on both sides,
because the penalty in the error rate for each double clicks is
$50\%$. Therefore, one has to use a lower mean photon number. This
effect decreases with the distance, and in the long distance limit
this advantage vanishes. Moreover, note that by 
adding the vacuum gain in the lower bound formula the resulting 
maximal achievable distance is shifted by around $\db=10$. 

We have shown that the detector decoy idea provides a simple method to
adapt a single photon security proof to its full optical
implementation, while still providing similar key rates as those
arising from a security proof using the squash model assumption. Its
main advantage relies on the fact that it can be straightforwardly
applied also to QKD protocols, like the active 6-state protocol, where
the squash model, the other ``adaption technique'', does not work.

\section{``Plug \& Play'' configuration}\label{sec_pp}

The main feature of the Plug \& Play configuration for QKD is that it is
intrinsically stable and polarization independent
\cite{stucki02a,muller97a,ribordy00a}. Apart from synchronization
between Alice and Bob, no further adjustments are necessary. This fact
renders this proposal a promising approach for commercial QKD systems.  

Specifically, in this type of QKD schemes Bob sends to Alice a train of bright
laser pulses through the quantum channel. On the receiving side, Alice first
attenuates the incoming signals to a suitable weak intensity. Afterwards, she
codes the secret key information using phase coding, and sends the resulting
weak pulses back to Bob, who detects them. The main idea behind this
bi-directional quantum communication design is that now the interferometers
used in a practical implementation of the scheme are self-stabilized because
the light passes through them twice. Moreover, if the reflection on AliceÕs
side is done by means of a Faraday mirror, then the polarization effects of
the quantum channel can also be compensated.  

A full security proof of a Plug \& Play system has recently been given in
Ref.~\cite{plug_play_active}. However, the security analysis contained
in Ref.~\cite{plug_play_active} is based on a slightly modification of
the hardware included in the original Plug \& Play proposal. In
particular, Alice performs three measures that enhance the security of
the protocol and also simplify its investigation
\cite{plug_play_active,fasel06}. First, she blocks any undesired
optical mode by means of an optical filter. Then, she performs active phase
randomization. This last action transforms the incoming signals into a
classical mixture of Fock states. Finally, she measures the photon number
distribution of the pulses received in order to estimate some bounds on the
photon number statistics of the output signals. This can be done by randomly
sampling the incoming pulses with an optical switch followed by an intensity
monitor (Case A in Fig.~\ref{plug_play}). The beam splitter that
appears in this figure is used to implement the decoy state method
which improves the whole performance of the scheme.  More recently, a
similar proposal has also been analyzed
\cite{plug_play_pasive}. Basically, it substitutes the optical switch
with a passive beam splitter  (Case B in Fig.~\ref{plug_play}).  
\begin{figure}
  \begin{center}
    \includegraphics[scale=0.66]{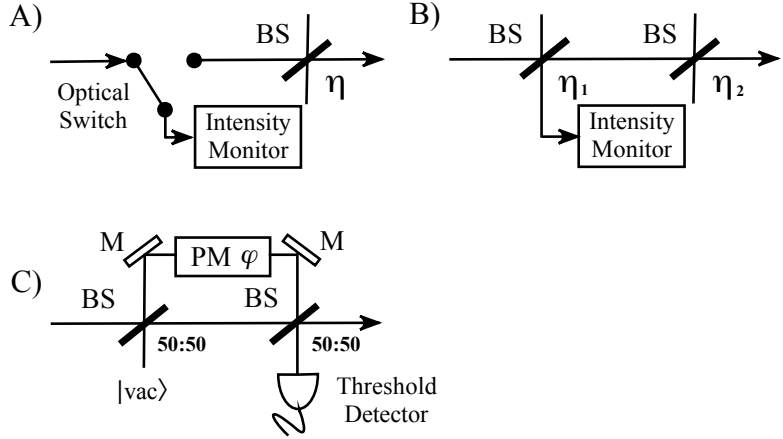}
  \end{center}
  \caption{Case A) Schematic diagram of the detection setup employed by Alice
    to estimate the photon number statistics of the output signals. The
    variable beam splitter which appears in the figure implements the
    decoy state method \cite{plug_play_active}. Case B) Illustration
    of a more recent experimental proposal to achieve the same goal
    \cite{plug_play_pasive}. It uses a passive beam splitter together
    with an intensity monitor. Like before, the second beam splitter in
    the figure is used to realize the decoy state method. Case C)
    Alternative method based on one balanced Mach-Zehnder
    interferometer combined with a threshold detector. PM denotes a
    phase modulator, M represents a mirror, and $\ket{\rm{vac}}$ is a
    vacuum state. \label{plug_play}} 
\end{figure}

In this section, we present very briefly an alternative experimental technique
to estimate the photon number statistics of the output signals. It is based on
the detector decoy idea presented in Sec.~\ref{seca}. Moreover, it allows
Alice to perform the decoy state method simultaneously, {\it i.e.}, without
using an additional variable beam splitter. The scheme is illustrated in
Fig.~\ref{plug_play} as Case C. It consists on a balanced Mach-Zehnder
interferometer combined with a threshold detector. After the active phase
randomization step performed by Alice, the signal states entering the
interferometer are given by Eq.~\ref{inputs}. Now, the probability that the
threshold detector does not click is given by   
\begin{equation}
  \label{prob_vacuum_plug}
  p_{\rm vac}(\varphi)=\sum_{n=0}^{\infty}
  \left(\frac{1-\cos\varphi}{2}\right)^np_n, 
\end{equation}
where $\varphi$ represents the phase imprinted by the phase modulator of the
interferometer. Like before, if Alice varies, independently and at random for
each signal, the phase $\varphi$ of her setup, then, from the observed
probabilities $p_{\rm vac}(\varphi)$, together with the knowledge of the
parameters $[(1-\cos\varphi)/2]^n$, she can estimate the photon number
distribution $p_n$ with high confidence. 
Given that the probabilities $p_n$ are now known, Alice
can also estimate the photon number statistics $q_n$ of the output
signals. These probabilities are given by
\begin{equation}
  q_n=\sum_{m=n}^{\infty}\frac{m!}{n! (n-m)!}\frac{p_m}{2^m}(1-\cos\varphi)^n
  (1+\cos\varphi)^{m-n}.  
\end{equation}
Note that by varying the phase $\varphi$ of her interferometer Alice also
modifies simultaneously the photon number probabilities $q_n$ of the
output signals, as required in the decoy state method, without the
need of an additional beam splitter to perform this task.

\section{Conclusion and outlook }\label{CONC}

In this paper we have analyzed a simple technique which allows the
direct application of a single photon security proof for quantum key
distribution (QKD) to its physical, full optical implementation. This
so-called detector decoy method is conceptually different to that of
the squash model, the other adaptation mechanism. It is based on an
estimation procedure for the photon number distribution of the
incoming light field that uses only a simple threshold detector in
combination with a variable attenuator. The detector decoy method is
similar in spirit to that of the usual decoy state technique: Since
the eavesdropper does not know the particular detection efficiency
setting used to measure the signals, any eavesdropping attempt must
leave the expected photon number distribution unchanged (similar to
the conditional channel losses in the decoy state
technique). 

Specifically, we have investigated an entanglement based QKD scheme
with an untrusted source where Alice and Bob actively choose the
measurement basis of either the BB84 or the 6-state protocol. The
security of both schemes is proven solely by means of the detector decoy
method and without the need of a squash model, which would have to be
proven to be correct for each measurement device anew. Besides, and
opposite to the squash model paradigm, the detector decoy technique
allows the legitimate users to discard double click events from the
raw key data. As a result, it turns out that the secret key rates in
the infinite (or sufficiently large) key rate limit of a BB84
simulated experiment are comparable with each other for both
alternatives. However, the detector decoy idea offers a slightly
better performance in those scenarios where there exists no squash
model, like in the 6-state 
protocol. In any real-life QKD experiment much more obstacles have to
be taken care of and thus the situation can change quite drastically,
mainly because of finite size effects. Nevertheless, for the current
increasing interest in examining the finite size behavior of
different protocols, it can only be of advantage to have a broader
spectrum of different proof techniques available, even if they all
show a similar behavior in the asymptotic key rate limit. 
Finally, as another potential application of the detector decoy method
in QKD, we have briefly described an experimental procedure to
estimate the photon number statistics of the output signals in a
``Plug \& Play'' QKD configuration. We believe there might be
many other potential applications of this method in QKD, like for
instance in Ref.~\cite{fung08a}. In addition, it could also be used to
estimate the single photon contribution in the two state protocol with
a strong reference pulse \cite{ben92,tamaki06suba}.   

To conclude, let us mention that there might be scenarios where it is
not really necessary to insert and vary the transmittance of an
additional beam splitter in the measurement device. For instance, let
us consider the efficient, passive BB84 measurement setup, in which a
beam splitter of high transmittance $\eta=1-\Delta$ splits the
incoming light in favor of one basis versus the other. With this
measurement apparatus, one can obtain directly three different beam
splitter settings to apply the detector decoy formalism: Using the
overall ``no click'' outcome of all detectors gives $\eta_1=1$,
whereas if one ignores all the outcomes of only one basis and looks at
the no click outcomes in the other basis, then one obtains two more
settings, $\eta_2=1-\Delta$ and $\eta_3=\Delta$. Although these three
settings are different from the ones given in Proposition~\ref{prop_n}
they can still provide good estimations of the single photon
contribution. Moreover, the method could be improved even
further. After all, in showing security we have not used all available
information from our measurement results, as the further occurrences of
double or multiclicks in our detection devices has been ignored. The
use of this extra knowledge can only enhance the estimation procedure
and thus can further reduce the number of necessary detector decoy settings. In
fact, it is possible to provide a BB84 security proof by just using an
estimation technique \cite{koashinew}. It might be interesting to
compare the detector decoy idea with the results presented in
Ref.~\cite{koashinew}, and we leave these open questions for further analysis.

\section*{Acknowledgements}

The authors wish to thank H.-K. Lo, W. Mauerer, K. Tamaki, X. Ma and
J. Lavoie for very useful discussions, and in particular M. Genovese
for pointing out important references. M. Curty especially thanks
N. L\"utkenhaus and H.-K. Lo for hospitality and support during his
stay at the Institute for Quantum Computing (University of Waterloo)
and at the University of Toronto, where this manuscript was 
finished. This work was supported by the European Projects SECOQC and
QAP, by the NSERC Discovery Grant, Quantum Works, CSEC, by Xunta de
Galicia (Spain, Grant No. INCITE08PXIB322257PR), and by University of
Vigo (Program ``Axudas \'a mobilidade dos investigadores''). 
 
\section*{References}

\bibliographystyle{iopart-num}

\providecommand{\newblock}{}

\end{document}